\documentclass[letterpaper,11pt]{article}

\usepackage[margin=1in]{geometry}  

\usepackage[margin=1in]{geometry}  

\usepackage{amsmath,amssymb,amsthm}

\usepackage{bbm,tikz}
\usepackage{authblk}

\usepackage{prettyref,xspace}
\newrefformat{eq}{(\ref{#1})}
\newrefformat{thm}{Theorem~\ref{#1}}
\newrefformat{th}{Theorem~\ref{#1}}
\newrefformat{chap}{Chapter~\ref{#1}}
\newrefformat{sec}{Section~\ref{#1}}
\newrefformat{algo}{Algorithm~\ref{#1}}
\newrefformat{fig}{Fig.~\ref{#1}}
\newrefformat{tab}{Table~\ref{#1}}
\newrefformat{rmk}{Remark~\ref{#1}}
\newrefformat{clm}{Claim~\ref{#1}}
\newrefformat{def}{Definition~\ref{#1}}
\newrefformat{cor}{Corollary~\ref{#1}}
\newrefformat{lmm}{Lemma~\ref{#1}}
\newrefformat{prop}{Proposition~\ref{#1}}
\newrefformat{pr}{Proposition~\ref{#1}}
\newrefformat{app}{Appendix~\ref{#1}}
\newrefformat{ques}{Question~\ref{#1}}

\newcommand{\distR}{\mathrm{dist}_R}

\newcommand{\pth}[1]{\left( #1 \right)}
\newcommand{\qth}[1]{\left[ #1 \right]}
\newcommand{\sth}[1]{\left\{ #1 \right\}}
\newcommand{\floor}[1]{{\left\lfloor {#1} \right \rfloor}}

\newcommand{\eg}{e.g.\xspace}
\newcommand{\ie}{i.e.\xspace}
\newcommand{\iid}{i.i.d.\xspace}

\newtheorem{theorem}{Theorem}
\newtheorem{lemma}{Lemma}

\newtheorem{prop}{Proposition}

\theoremstyle{definition}

\newtheorem{remark}{Remark}
\newtheorem{question}{Question}
\newcommand{\lk} { \lim_{k \rightarrow \infty}}

\title{Equivalence of additive-combinatorial linear inequalities for Shannon entropy and differential entropy}
\author{Ashok Vardhan Makkuva and Yihong Wu\thanks{Ashok Vardhan Makkuva is with
the Department of ECE and the Coordinated Science Lab, University of Illinois at Urbana-Champaign, Urbana, IL, email: \texttt{makkuva2@illinois.edu}. Yihong Wu is with the Department of Statistics, Yale University, New Haven CT 06511,  email: \texttt{yihong.wu@yale.edu}.}}
\date{\today}

\usepackage[
            CJKbookmarks=true,
            bookmarksnumbered=true,
            bookmarksopen=true,
            colorlinks=true,
            citecolor=red,
            linkcolor=blue,
            anchorcolor=red,
            urlcolor=blue
            ]{hyperref}
  
  \renewcommand{\footnoterule}{%
  \kern -3pt
  \hrule width \textwidth height 1pt
  \kern 2pt
}

\newcommand{\Z}{\mathbb{Z}}

\newcommand{\TV}{d_\mathrm{TV}}

\newcommand{\Mod}[1]{\ (\text{mod}\ #1)}
\newcommand{\add}{\sum_{j=1}^m}

\newcommand{\torus}{\mathbb{T}}
\newcommand{\diverge}{\to\infty}
\newcommand{\reals}{\mathbb{R}}
\newcommand{\naturals}{\mathbb{N}}
\newcommand{\integers}{\mathbb{Z}}
\newcommand{\Expect}{\mathbb{E}}

\newcommand{\prob}[1]{\mathbb{P}\left[#1\right]}
\newcommand{\pprob}[1]{\mathbb{P}[#1]}
\newcommand{\calU}{\mathcal{U}}

\begin{document}
\maketitle


\begin{abstract}
This paper addresses the correspondence between linear inequalities of Shannon entropy and differential entropy for sums of independent group-valued random variables.
We show that any balanced (with the sum of coefficients being zero) linear inequality of Shannon entropy holds if and only if its differential entropy counterpart also holds; moreover, any linear inequality for differential entropy must be balanced.
In particular, our result shows that recently proved differential entropy inequalities by Kontoyiannis and Madiman \cite{KM14} can be deduced from their discrete counterparts due to Tao \cite{Tao10} in a unified manner. Generalizations to certain abelian groups are also obtained. 

Our proof of extending inequalities of Shannon entropy to differential entropy relies on a result of R\'enyi \cite{Renyi59} which relates the Shannon entropy of a finely discretized random variable to its differential entropy and also helps in establishing the entropy of the sum of quantized random variables is asymptotically equal to that of the quantized sum; the converse uses the asymptotics of the differential entropy of convolutions with weak additive noise.

\end{abstract}

\tableofcontents

\section{Introduction and main result}
\subsection{Additive-combinatorial inequalities for cardinality and Shannon entropy}
Over the past few years, the field of additive combinatorics has invited a great deal of mathematical activity; see \cite{TV06} for a broad introduction. An important repository of tools in additive combinatorics is the sumset inequalities, relating the cardinalities of the sumset and the difference set 
$A\pm B=\{a\pm b:a\in A,b\in B\}$ 
to those of $A$ and $B$, where $A$ and $B$ are arbitrary subsets of integers, or more generally, any abelian group.

 One can consider the information-theoretic analogs of these additive combinatoric inequalities by replacing the sets by (independent, discrete, group-valued) random variables and, correspondingly, the log-cardinality by the Shannon entropy. For example, the inequality 
 \[
 \max\{|A|,|B|\} \leq |A+B| \leq |A||B|
 \]
  translates to 
\begin{equation}
\max\left\{H\left(X\right),H\left(Y\right)\right\} \leq H\left(X+Y\right)\leq H\left(X\right)+H\left(Y\right),	
	\label{eq:HHH}
\end{equation}
  which follows from elementary properties of entropy. The motivation to consider these analogs comes from the interpretation that the Shannon entropy 
  \[
  H\left(X\right)\triangleq \sum_x \prob{X=x} \log \frac{1}{\prob{X=x}}
  \]
   of a discrete random variable $X$ can be viewed as the logarithm of the  \emph{effective cardinality} of the alphabet of $X$ in the sense of asymptotic equipartition property (AEP) \cite{cover}, which states that the random vector consisting of $n$ independent copies of $X$ is concentrated on a set of cardinality $\exp(n(H(X)+o(1))$ as $n\diverge$.
 While this observation was fruitful in deducing certain entropy inequality, \eg, Han's inequality \cite{Han78}, directly from their set counterpart cf.~\cite[p.~5]{Ruzsa09}, it has not proven useful for inequalities dealing with sums since the typical set of sums can be exponentially larger than sums of individual typical sets.
Forgoing this soft approach and capitalizing on the submodularity property of entropy,
 in the past few years several entropy inequalities for sums and differences have been obtained \cite{TV05,Lape,Madiman08,Tao10,MK10,MMT10}, such as the sum-difference inequality \cite[Eq.~(2.2)]{Tao10}
 \begin{equation}
	H(X+Y) \leq 3 H(X-Y) - H(X)-H(Y),
	\label{eq:sumdiff}
\end{equation}
which parallels the following (cf., \eg, \cite[Eq.~(4)]{GHR07})
\[
|A+B| \leq \frac{|A-B|^3}{|A||B|}.
\]
 More recently, a number of entropy inequalities for integer linear combinations of independent random variables
 have been obtained in \cite[Appendix E]{DoF.IT}, \eg,
\[
	H(p X + q Y) - H(X+Y) \leq (7 \floor{\log |p|} + 7 \floor{\log |q|} + 2) (2 H(X+Y) - H(X) - H(Y)),
\]
for non-zero integers $p,q$, which are counterparts of results on sum of dilated sets in \cite{bukh08}.

It is worth noting that all of the aforementioned results for Shannon entropy are \emph{linear inequalities} for entropies of weighted sums of independent random variables, which are of the general form:
\begin{equation}
\sum_{i=1}^n \alpha_i H \pth{\sum_{j=1}^m a_{ij} Z_j} \leq 0,
	\label{eq:linear-H}
\end{equation}
 with $a_{ij} \in \integers$, $\alpha_i \in \reals$, $Z_1,\ldots,Z_m$ being independent discrete group-valued random variables.

\subsection{Equivalence of Shannon and differential entropy inequalities}

Recall that the \emph{differential entropy} of a real-valued random vector $X$ with probability density function (pdf) $f_X$ is defined as 
\[
h\left(X\right)=\int f_X(x) \log \frac{1}{f_X(x)} dx.
\]
 Again, in the sense of AEP, $h(X)$ 
can be interpreted as the log-volume of the effective support of $X$ \cite{cover}.
In a similar vein, one can consider similar additive-combinatorial inequalities for differential entropies on Euclidean spaces. 
Recently Kontoyiannis and Madiman \cite{KM14} and Madiman and Kontoyiannis \cite{MK10,MK15} made important progress in this direction by 
showing that while the submodularity property, the key ingredient for proving discrete entropy inequalities, fails for differential entropy, 
several linear inequalities for Shannon entropy nevertheless extend \emph{verbatim} to differential entropy; for example, the sum-difference inequality \prettyref{eq:sumdiff} admits an exact continuous analog \cite[Theorem 3.7]{KM14}:
\begin{equation}
	h(X+Y) \leq 3 h(X-Y) - h(X)-h(Y).
	\label{eq:sumdiff-h}
\end{equation}
These results prompt us to ask the following question, which is the focus of this paper:



\begin{question}
\label{ques:main}
		Do all linear inequalities of the form \prettyref{eq:linear-H} for discrete entropy extend to differential entropies, and vice versa?
		\end{question}

%

 A simple but instructive observation reveals that all linear inequalities for differential entropies are always \emph{balanced}, that is, the sum of all coefficients must be zero. In other words, should 
\begin{equation}
\sum_{i=1}^n \alpha_i h \pth{\sum_{j=1}^m a_{ij} Z_j} \leq 0,	
	\label{eq:linear-h}
\end{equation}
hold for all independent $\reals^d$-valued $Z_j$'s, then we must have $\sum_{i=1}^n \alpha_i=0$. To see this, recall the fact that $h(aZ) = h(Z) + d \log a$ for any $a>0$; in contrast, Shannon entropy is scale-invariant. Therefore, whenever the inequality \prettyref{eq:linear-h} is unbalanced, \ie,  
$\sum_{i=1}^n \alpha_i \neq 0$, scaling all random variables by $a$ and sending $a$ to either zero or infinity leads to a contradiction.
For instance, in \prettyref{eq:HHH}, the left inequality (balanced) extends to differential entropy but the right inequality (unbalanced) clearly does not.

Surprisingly, as we show in this paper, a balanced linear inequality holds for Shannon entropy if and only if it holds for differential entropy, thereby fully resolving \prettyref{ques:main}.  
This result, in a way, demystifies the striking parallel between discrete and continuous entropy inequalities.
In particular, it shows that the results in \cite{KM14,MK15}, which are linear inequalities for mutual information such as $I(X;X+Y)=h(X+Y)-h(Y)$ or Ruzsa distance $\distR(X,Y) \triangleq h(X-Y)- \frac{1}{2}h(X)-\frac{1}{2}h(Y)$ \cite{Ruzsa09,Tao10,KM14}) and hence expressible as balanced linear inequalities for differential entropy, can be deduced from their discrete counterparts \cite{Tao10} in a unified manner. 
While our results establish that all balanced linear inequalities for Shannon entropy extend to differential entropy and vice versa, it is worth pointing out that this does not hold for affine inequalities. Note that non-trivial \emph{affine} inequality for Shannon entropy does not exist simply because one can set all random variables to be deterministic; however, this is not the case for differential entropy. For instance, the following balanced affine inequality
\begin{align}
h(X+Y) \geq \frac{1}{2} \pth{h(X)+h(Y)}+\frac{d}{2} \log 2 \label{eq:Exception}
\end{align}
holds for any independent $\reals^d$-valued random variables $X$ and $Y$, which is a direct consequence of the entropy power inequality (see \cite[Lemma 3.1]{Barron84} for generalizations of \prettyref{eq:Exception}). However, the Shannon entropy analogue of \prettyref{eq:Exception}, replacing all $h$ by $H$, is clearly false (consider deterministic $X$ and $Y$).On the other hand, there exists no unbalanced linear inequality for differential entropy while it's not true for Shannon entropy. Consider for instance, the Shannon entropy inequality
\begin{align*}
H(X+Y) \leq H(X)+H(Y) 
\end{align*}
holds for any independent discrete random variables $X$ and $Y$, which follows directly from the elementary properties of Shannon entropy. However, the differential entropy counterpart, $h(X+Y) \leq h(X)+h(Y)$ can be shown to be false by taking $X$ and $Y$ to be independent Gaussian random variables with zero mean and variance $\frac{1}{2\pi e}$ and $1$ respectively.

To explain our proof that discrete entropy inequalities admit continuous counterparts, we first note that the main tool for proving differential entropy inequalities in 
\cite{MK10,KM14,MK15} is the data processing inequality of mutual information, replacing the \emph{submodularity} of Shannon entropy exploited in \cite{Tao10}.
However, this method has been applied on a case-by-case basis as there seems to be no principled way to recognize the correct data processing inequality that needs to be introduced. 
Instead, to directly deduce a differential inequality from its discrete version, our strategy is to rely on a result due to R\'enyi \cite{Renyi59} which gives the asymptotic expansion of the Shannon entropy of a finely quantized continuous random variable in terms of its differential entropy, namely, 
\begin{equation}
H(\floor{m X}) =  d \log m + h(X) + o(1), \quad m\diverge	
	\label{eq:renyi0}
\end{equation}
 for continuous $\reals^d$-valued $X$.
In fact, this approach has been discussed in \cite{KM14} at the suggestion of a reviewer, where it was noted that differential entropy inequalities can be approximately obtained from their discrete counterparts via this quantization approach, since $H(\floor{m X} + \floor{m Y})$ and $H(\floor{m (X + Y)})$ can only differ by a few bits, which might be further improvable. Indeed, as we shall prove later  in \prettyref{lmm:implmm1}, this entropy difference is in fact vanishingly small, which enables the additive-combinatorial entropy inequalities to carry over exactly from discrete to Euclidean spaces, and, even more generally, for connected abelian Lie groups. 
Interestingly, in addition to bridging the discrete and continuous notion of entropy, R\'enyi's result also plays a key role in establishing the vanishing entropy difference.

In establishing that all linear discrete entropy inequalities follow from their continuous analogs, the following are the two key ideas of our approach: First we show that given any finite collection of discrete $\reals^d$-valued random variables, we can embed them into a high dimensional Euclidean space and project them back to $\reals^d$ such that the Shannon entropy of any linear combinations of the projected random variables is equal to an arbitrarily large multiple of that the given random variables. Next we add independent noise, e.g., Gaussian, with arbitrarily small variance to these projected discrete random variables and relate their Shannon entropy to the differential entropy of their noisy versions. 
Sending the variance to zero and then the dimension to infinity yields the desired inequality for discrete entropy.

\subsection{Main results}
Throughout the rest of the paper, to make the statements concise and exclude trivial cases, all differential entropies are assumed to exist and be finite. We now state our main results on linear entropy inequalities.

\begin{theorem}
\label{thm:main}
Let $\left(a_{ij}\right) \in \Z^{n \times m}$ satisfies that $a_{i1},\ldots,a_{im}$ are relatively prime, for each $i=1,\ldots,n$. Let $\alpha_1,\ldots,\alpha_n \in \mathbb{R}$ be such that $\sum_{i=1}^n \alpha_i=0$. Suppose for any independent $\integers^d$-valued random variables $U_1,\ldots,U_m$, the following holds:
\begin{align}
\sum_{i=1}^n \alpha_i H \pth{\sum_{j=1}^m a_{ij} U_j} \leq 0.
\label{eq:A}
\end{align}
Then for any independent $\reals^d$-valued continuous random variables $X_1,\ldots,X_m$, the following holds:
\begin{align}
\sum_{i=1}^n \alpha_i h \pth{\sum_{j=1}^m a_{ij} X_j} \leq 0
\label{eq:B}
\end{align}
\end{theorem}
\begin{remark}
Without loss of any generality, we can always assume that the coefficients of each linear combination of random variables in \prettyref{eq:A}
are relatively prime. This is because for each $i$ we can divide $a_{i1},\ldots,a_{im}$ by their greatest common divisor so that the resulting entropy inequality remains the same, thanks to the scale invariance of the Shannon entropy.
\end{remark}

\begin{theorem}
\label{thm:converse}
Let $(a_{ij}) \in \reals^{n \times m}$ and $\alpha_1,\ldots,\alpha_n \in \reals$ be such that $\sum_{i=1}^n \alpha_i =0$. If 
\begin{align*}
\sum_{i=1}^n \alpha_i h \pth{\sum_{j=1}^m a_{ij} X_j} \leq 0
\end{align*}
holds for any $\reals^d$-valued independent and continuous random variables $X_1,\ldots,X_m$, then
\begin{align*}
\sum_{i=1}^n \alpha_i H \pth{\sum_{j=1}^m a_{ij} U_j} \leq 0
\end{align*}
holds for any $\reals^d$-valued independent and discrete random variables $U_1,\ldots,U_m$.
\end{theorem}

\begin{remark}[iid random variables]
\label{rmk:allowiid}
For additive-combinatorial entropy inequalities, when (some of) the random variables are further constrained to be identically distributed, a number of strengthened inequalities have been obtained. For instance, if $U$ and $U'$ are independent and identically distributed (iid) discrete random variables, then (cf., \eg, \cite[Theorems 1.1 and 2.1]{MK10})
\begin{equation} 
	\frac{1}{2}\leq \frac{H(U-U')-H(U)}{H(U+U')-H(U)} \leq 2
	\label{eq:doubling-H}
\end{equation}
and for iid continuous $X,X'$,
\begin{equation}
	\frac{1}{2}\leq \frac{h(X-X')-h(X)}{h(X+X')-h(X)} \leq 2
	\label{eq:doubling-h}
\end{equation}
which are stronger than what would be obtained from \prettyref{eq:sumdiff} and \prettyref{eq:sumdiff-h} by substituting $Y=X'$.

As evident from the proof, both \prettyref{thm:main} and \prettyref{thm:converse} apply verbatim to entropy inequalities involving independent random variables with arbitrary distributions.
Consequently, \prettyref{eq:doubling-h} and \prettyref{eq:doubling-H} are in fact equivalent. Formally, fix a partition $S_1, \ldots, S_K$ of $[m] \triangleq \{1,\ldots,m\}$. Then \prettyref{eq:A} holds for independent $U_1,\ldots,U_m$ so that $\{U_j\}_{j\in S_k}$ are iid for $k\in [K]$ if and only if
\prettyref{eq:B} holds for independent $X_1,\ldots,X_m$ so that $\{X_j\}_{j\in S_k}$ are iid for $k\in [K]$. 
It is worth noting that this result is not a special case of Theorems \ref{thm:main} and \ref{thm:converse}; nevertheless, the proofs are identical. 
\end{remark}

\begin{remark}
\label{rmk:chan}	
The nature of the equivalence results that we obtained in this paper for linear inequalities for weighted sums of independent random variables bear some similarity to a result established by Chan in \cite{Chan03} for linear entropy inequalities of \emph{subsets} of random variables, as opposed to sums of independent random variables. In particular, he established that the class of linear inequalities for Shannon entropy and differential entropy are equivalent provided the inequalities are ``balanced" in the following sense.
For example, consider the following entropy inequalities for discrete random variables $X_1$ and $X_2$:
\begin{align}
H(X_1)+H(X_2) - H(X_1,X_2) \geq 0, \label{eq:chan1}\\
H(X_1,X_2) - H(X_1) \geq 0 \label{eq:chan2}.
\end{align}
The inequality \prettyref{eq:chan1} is said to be balanced because the sum of the coefficients of the entropy terms in which $X_1$ appears equals zero and the same is true for $X_2$ as well. However, the inequality \prettyref{eq:chan2} is unbalanced because $X_2$ appears only in the first term.
Though the notion of balancedness considered in \cite{Chan03} is different from ours, the technique employed for extending the discrete entropy inequalities to the continuous case is similar to ours, i.e., through discretization of continuous random variables; however, as discussed before, the key argument is to show that the entropy of the sum of quantized random variables is asymptotically equal to that of the quantized sum, a difficulty which is not present in dealing with subsets of random variables.

To deduce the discrete inequality from its continuous counterpart, the method in \cite{Chan03} is to assume, without loss of generality, the discrete random variables are integer-valued
and use the fact that $H(A) = h(A+U)$ for any $\integers$-valued $A$ and $U$ independently and uniformly distributed on $[0,1]$.
Clearly this method does not apply to sums of independent random variables. 
\end{remark}

\subsection{Organization}
	\label{sec:org}
	The rest of the paper is organized as follows.
	Before giving the proof of the main results,
 in \prettyref{sec:sharp} we pause to discuss the open problem of determining the sharp constants in additive-combinatorial entropy inequalities and the implications of our results.
	 The proof of the main theorems are given in Sections \ref{sec:pf-main} and \ref{sec:pf-converse}, with the technical lemmas proved in \prettyref{sec:pflmm}.
Following \cite{KM14}, the notion of differential entropy can be extended to locally compact groups by replacing the reference measure (Lebesgue) by the corresponding Haar measure. In \prettyref{sec:groupstuff} we generalize \prettyref{thm:main} to random variables taking values in connected abelian Lie groups.

\section{On sharp constants in additive-combinatorial entropy inequalities}
	\label{sec:sharp}

The entropy inequalities \prettyref{eq:doubling-H} and \prettyref{eq:doubling-h} can be viewed as the information theoretic analogs of the following additive-combinatorial inequality proved by Ruzsa \cite{Ruzsa91}: For any finite $A \subset \integers^n$( or any abelian group)
\begin{align}
\label{eq:setRuzsa}
\log \frac{|A-A|}{|A|} \leq 2 \log \frac{|A+A|}{|A|}.
\end{align}
The constant $``2"$ in \prettyref{eq:setRuzsa} is known to be sharp (see \cite{HRY99} or \cite[p.~107]{Ruzsa-lecture}).
The crucial idea for the construction is to approximate cardinality by volume by considering the lattice points inside a convex body.
In particular, for any convex body $K$ in $\reals^n$, denote its quantized version $\qth{K}_L \triangleq K \cap (\frac{1}{L} \integers^n)$, where $L\in\naturals$. The sum and difference sets of $\qth{K}_L$ is related to those of $K$ through $\qth{K\pm K}_L =\qth{K}_L \pm \qth{K}_L$.
If we fix the dimension $n$ and let $L \rightarrow \infty$, it is well-known that the cardinality of $\qth{K}_L$ is related to the volume of $K$ via $
	|[K]_L| = \text{vol}(K) L^n (1+o(1))$.
Thus,
\begin{align*}
	\frac{|[K]_L\pm [K]_L|}{|[K]_L|} = \frac{\text{vol}(K\pm K)}{\text{vol}(K)} (1+o(1)).
\end{align*}
A classical result of Rogers and Shephard \cite{RS57} states that for any convex $K \in \reals^n$, $\text{vol}(K-K) \leq {2n \choose n}\text{vol}(K)$ with  
equality if and only if $K$ is a simplex. Since $K$ is convex, $K+K=2K$ and thus $\text{vol}(K+K)=2^n \text{vol}(K)$. Now taking $K$ to be the standard simplex
$\Delta_n = \sth{x \in \reals^n_+: \sum_{i=1}^n x_i \leq 1}$, we obtain
\begin{align*}
\frac{\log \frac{|[\Delta_n]_L- [\Delta_n]_L|}{|[\Delta_n]_L|}}{\log \frac{|[\Delta_n]_L+ [\Delta_n]_L|}{|[\Delta_n]_L|}}	
	 = \frac{\log \frac{{2n \choose n}}{n!}-\log \frac{1}{n!}+o_L(1)}{\log \frac{2^n}{n!}-\log \frac{1}{n!} +o_L(1)}= \frac{\log \binom{2n}{n}+o_L(1)}{n \log 2 +o_L(1)},
\end{align*}
where we used $\text{vol}(\Delta_n)=\frac{1}{n!}, \text{vol}(\Delta_n-\Delta_n)=\frac{1}{n!}{2n \choose n}$ and $\text{vol}(\Delta_n+\Delta_n)=\frac{2^n}{n!}$. Sending $L \rightarrow \infty$ followed by $n \rightarrow \infty$ yields that the sharpness of \prettyref{eq:setRuzsa}.

Analogously, one can investigate the best possible constants in the Shannon entropy entropy inequalities \prettyref{eq:doubling-H} as well as its continuous analog \prettyref{eq:doubling-h}. It is unclear if the constants $1/2$ and $2$ are the best possible. However, as a consequence of \prettyref{thm:main} and \prettyref{thm:converse}, one can establish that the sharp constants for the discrete and continuous versions are the same, and dimension-free (see \prettyref{app:sharp} for a proof):
\begin{prop}
\label{prop:sharp}	
	For \iid $U$ and $U'$ and \iid $X$ and $X'$, 
\begin{align*}
\frac{1}{2} \leq & \inf_{U \in \integers^n} \frac{H(U-U')-H(U)}{H(U+U')-H(U)} = \inf_{X \in \reals^n} \frac{h(X-X')-h(X)}{h(X+X')-h(X)} \\
\leq & \sup_{X \in \reals^n} \frac{h(X-X')-h(X)}{h(X+X')-h(X)} =\sup_{U \in \integers^n} \frac{H(U-U')-H(U)}{H(U+U')-H(U)} \leq 2. \label{eq:Hhsame}
\end{align*}
Furthermore, the infimum and the supremum are independent of the dimension $n$.
\end{prop}
It is worth pointing out that the dimension-freeness of the best Shannon entropy ratio follows from standard arguments (tensorization and linear embedding of $\integers^n$ into $\integers$), which have been previously used for proving analogous results for set cardinalities \cite{HRY99}; however, it is unclear how to directly prove the ratio of differential entropy is dimension-independent without resorting to \prettyref{thm:main}.
In view of the success of continuous approximation in proving the sharpness of \prettyref{eq:setRuzsa}, proving the sharpness of \prettyref{eq:doubling-h} for differential entropies might be more tractable than its discrete counterpart \prettyref{eq:doubling-H}.

\section{Proof of \prettyref{thm:main}}
\label{sec:pf-main}
We first introduce the notations followed throughout the paper.
For $x\in\reals$, let $\lfloor x \rfloor \triangleq \max\{k \in \integers: k \leq x\}$ and  $\{x\}=x-\floor{x}$ denote its integer and fractional parts, respectively. For any $k\in\naturals$, define
\begin{equation}
	\left[x\right]_k \triangleq \frac{\lfloor 2^kx \rfloor }{2^k}, \quad \left\{x\right\}_k \triangleq\frac{\{2^kx\}}{2^k}.
	\label{eq:quant}
\end{equation}
Hence,
\begin{align*}
x= \frac{\lfloor 2^kx \rfloor}{2^k} + \frac{\{2^kx\}}{2^k}=\left[x\right]_k+\{x\}_k.
\end{align*}
For $x \in \mathbb{R}^d$, $\left[x\right]_k$ and $\left\{x\right\}_k$ are defined similarly by  applying the above operations componentwise.

For $N>0$, denote the hypercube $B_N^{(d)} \triangleq \left[-N,N\right]^d$. For a $\mathbb{R}^d$-valued random variable $X$, let $X^{(N)}$ denote a random variable distributed according to the conditional distribution $P_{X | {X \in B_N^{(d)}}}$.
If $X$ has a pdf $f_X$, then $X^{(N)}$ has the following pdf:
\begin{align}
f_{X^{(N)}}(x)= \frac{f_{X}(x) \mathbbm{1}\{x \in B_N^{(d)}\}}{\pprob{X\in B_N^{(d)}}}.
\label{eq:ftruncate}
\end{align}


The following lemma is the key step to proving \prettyref{thm:main}.
\begin{lemma}
\label{lmm:implmm1}
Let $X_1,\ldots,X_m$ be independent $\left[0,1\right]^d$-valued continuous random variables such that   both $h\left(X_j\right)$ and $H\left(\lfloor X_j \rfloor \right)$ are finite for each $j \in \left[m\right]$. Then for any $a_1,\ldots,a_m \in \Z$ that are relatively prime, 
$$
\lk \left( H\bigg( \left[ \sum_{i=1}^m a_iX_i\right]_k  \bigg)- H\bigg(\sum_{i=1}^m a_i \left[X_i \right]_k \bigg)   \right) =0.
$$
\end{lemma}

The next lemma allows us to focus on bounded random variables.
\begin{lemma}[Truncation]
\label{lmm:implmm2}
Let $X_1,\ldots,X_m$ be independent $\mathbb{R}^d$-valued random variables and $a_1,\ldots,a_m \in \mathbb{R}$.
If each $X_j$ has an absolutely continuous distribution and $h(X_j)$ is finite, then
\begin{align*}
\lim_{N\diverge} h \left( \sum_{j=1}^m a_j X_j^{(N)}\right) = h\left(\sum_{j=1}^m a_j X_j \right).
\end{align*}

\end{lemma}

%

The following lemma is a particularization of \cite[Theorem 1]{Renyi59} (see \prettyref{eq:renyi0}) to the dyadic subsequence $m=2^k$:
\begin{lemma}
\label{lmm:lmmRenyi}
For any $\mathbb{R}^d$-valued random variable $X$ with an absolutely continuous distribution such that both $H\left( \lfloor X \rfloor \right)$ and $h\left(X\right)$ are finite,
$$
\lk \left(H \left(\left[X\right]_k \right) -dk\log2\right)= h\left(X\right).
$$
\end{lemma}

We are now ready to prove \prettyref{thm:main}.
\begin{proof}
We start by considering the case where $X_j \in \left[0,1\right]^d$ for each $j \in \left[m\right]$.  Since $X_j$'s are independent and $2^k\left[X_j\right]_k$ is $\mathbb{Z}^d$-valued for each $j \in \left[m\right]$, by assumption,
\begin{align}
\sum_{i=1}^n \alpha_i H \left( \sum_{j=1}^m a_{ij} \left[X_j \right]_k \right) \leq 0 \label{eq:star1}
\end{align}
holds where 
\begin{align}
 \sum_{i=1}^n \alpha_i=0. \label{eq:balanced}
\end{align}
By \prettyref{lmm:lmmRenyi}, 
$ H \left( \left[ X \right]_k \right)= dk \log2 + h \left( X \right)+ o_k(1)$. Thus,
\begin{align*}
h \left(  \sum_{j=1}^m a_{ij} X_j   \right)+ d k \log2 +o_k(1) &= H \left( \left[\sum_{j=1}^m a_{ij}X_j \right]_k   \right)\\
& \stackrel{(a)}= H \left( \sum_{j=1}^m a_{ij} \left[X_j \right]_k \right)+ o_k(1),
\end{align*}
where (a) follows from \prettyref{lmm:implmm1}.
Multiplying on both sides by $\alpha_i$ and summing over $i$, and in view of \eqref{eq:balanced}, we have
\begin{align}
\nonumber
 \sum_{i=1}^n \alpha_i h \left(  \sum_{j=1}^m a_{ij} X_j   \right)+ o_k(1)  &= \sum_{i=1}^n \alpha_i H \left( \sum_{j=1}^m a_{ij} \left[X_j \right]_k \right).
\end{align}
By \eqref{eq:star1}, sending $k$ to infinity yields the desired result. 

For the general case where $X_j \in \mathbb{R}^d$, let $Y_i= \sum_{j=1}^m a_{ij}X_j$ for $i \in \left[n\right]$. 
Let $\tilde{X}_j^{(N)} \triangleq \frac{X_j^{(N)}+N}{2N}$, which belongs to $\left[0,1\right]^d$. Thus,
\begin{align}
\nonumber
\sum_{i=1}^n \alpha_i h \left( \sum_{j=1}^m a_{ij} \tilde{X}_j^{(N)} \right) &= \sum_{i=1}^n \alpha_i h \left( \sum_{j=1}^m a_{ij}X_j^{(N)}  \right) + \sum_{i=1}^n \alpha_i\cdot \log \left(\frac{1}{2N}\right)^d\\
& = \sum_{i=1}^n \alpha_i h \left( \sum_{j=1}^m a_{ij} {X}_j^{(N)} \right), \label{eq:detgone}
\end{align}
where \eqref{eq:detgone} follows from \eqref{eq:balanced}. Hence, 
\begin{align*}
 \sum_{i=1}^n \alpha_i h\left( Y_i\right) & \stackrel{(a)}= \lim_{N \rightarrow \infty} \sum_{i=1}^n \alpha_i h \left( \sum_{j=1}^m a_{ij} X_j^{(N)} \right)\\
&\stackrel{(b)}=  \lim_{N \rightarrow \infty}\sum_{i=1}^n \alpha_i h \left( \sum_{j=1}^m a_{ij} \tilde{X}_j^{(N)} \right)\\
& \stackrel{(c)}\leq 0,
\end{align*}
where $(a)$ follows form \prettyref{lmm:implmm2} and $(b)$ follows from \eqref{eq:detgone}, and $(c)$ follows from the earlier result for $\left[0,1\right]^d$-valued random variables. The proof of \prettyref{thm:main} is now complete.
\end{proof}

\section{Proof of \prettyref{thm:converse}}
\label{sec:pf-converse}
\prettyref{thm:converse} relies of the following two lemmas. The first result is a well-known asymptotic expansion of the differential entropy of a discrete random variable contaminated by weak additive noise.
For completeness, we provide a short proof in \prettyref{sec:NoisyMutual}.

\begin{lemma}
\label{lmm:NoisyMutual}
Let $U$ be a discrete $\reals^d$-valued random variable such that $H(U) < \infty$ and $Z$ be a $\reals^d$-valued continuous random variable with $h(Z) > - \infty$. If $U$ and $Z$ are independent, then 
\begin{align*}
h(U+\varepsilon Z)=h(Z)+\log \varepsilon + H(U)+o_\varepsilon(1).
\end{align*}
\end{lemma}

The following lemma, proved in \prettyref{sec:Embedding}, allows us to blow up the Shannon entropy of linear combinations of discrete random variables arbitrarily.  

\begin{lemma}
\label{lmm:Embedding}
Let $U_1,\ldots,U_m$ be $\reals^d$-valued discrete random variables. Let $k \in \mathbb{N}$. Then for any $A=(a_{ij}) \in \reals^{n \times m}$, there exist $\reals^d$-valued discrete random variables $U_1^{(k)},\ldots,U_m^{(k)}$ such that
\begin{align*}
H \pth{\sum_{j=1}^m a_{ij} U_j^{(k)}} = k H \pth{\sum_{j=1}^m a_{ij} U_j}, \forall i \in [n].
\end{align*}
\end{lemma}

We now prove \prettyref{thm:converse}.

\begin{proof}
Let $Z_j$ be independent $\reals^d$-valued Gaussian random variables with zero mean and $U_1,\ldots,U_m$ be independent $\reals^d$-valued discrete random variables. Let $U_1^{(k)},\ldots,U_m^{(k)}$ be independent $\reals^d$-valued discrete random variables such that $H\pth{\sum_{j=1}^m a_{ij}U_j^{(k)}}=kH\pth{\sum_{j=1}^m a_{ij}U_j}$ for each $i \in [n]$, guaranteed by \prettyref{lmm:Embedding}.

%
%


Let $\varepsilon>0$. For each $j \in [m]$, let $X_j=U_j^{(k)}+\varepsilon Z_j$. Then we have,
\begin{align*}
h\pth{X_{j}}= H(U_{j}^{(k)})+ h(Z_j)+\log \varepsilon +o_{\varepsilon}(1).
\end{align*}
Hence, for each $i \in [n]$,
\begin{align*}
h \pth{\sum_{j=1}^m a_{ij}X_{j}} &= h\pth{\sum_{j=1}^m a_{ij} U_{j}^{(k)} +\varepsilon  \sum_{j=1}^m a_{ij} Z_j} \\
& \stackrel{(a)}= H \pth{\add a_{ij} U_{j}^{(k)} } + h\pth{\add a_{ij}Z_j} +\log \varepsilon +o_\varepsilon(1)\\
&=k H \pth{\add a_{ij} U_j}+h\pth{\add a_{ij}Z_j} +\log \varepsilon +o_\varepsilon(1),
\end{align*}
where $(a)$ follows from \prettyref{lmm:NoisyMutual}.
Since $X_j$'s are independent, by assumption, $\sum_{i=1}^n \alpha_i h\pth{\add a_{ij}X_{j}} \leq 0$ where $\sum_{i=1}^n \alpha_i$. Hence,
\begin{align*}
k \sum_{i=1}^n \alpha_i H\pth{\add a_{ij} U_j}+ \sum_{i=1}^n \alpha_i h\pth{\add a_{ij} Z_j} + o_\varepsilon(1) \leq 0.
\end{align*}
Thus,
\begin{align*}
\sum_{i=1}^n \alpha_i H \pth{\add a_{ij}U_j} + \frac{ \sum_{i=1}^n \alpha_i h \pth{\add a_{ij}Z_j}}{k}+ \frac{o_\varepsilon(1)}{k} \leq 0.
\end{align*}
The proof is completed by letting $\varepsilon \rightarrow 0$ followed by $k \rightarrow \infty$.
\end{proof}

\section{Proofs of lemmas}
\label{sec:pflmm}

\subsection{Proof of \prettyref{lmm:implmm1}}
Let $a_1,\ldots,a_m \in \Z$ and $X_1,\ldots,X_m$ be $\mathbb{R}^d$-valued random variables. Then
\begin{align*}
\left[\sum_{i=1}^m a_i X_i \right]_k  
&= \frac{\Big \lfloor 2^k \sum_{i=1}^m a_i X_i \Big \rfloor}{2^k} =\frac{ \Big \lfloor \sum_{i=1}^m a_i \lfloor 2^kX_i \rfloor + \lfloor \sum_{i=1}^m a_i \{ 2^kX_i \}  \Big \rfloor}{2^k} \\
&= \sum_{i=1}^m a_i [X_i]_k +  \frac{\lfloor \sum_{i=1}^m a_i \{ 2^kX_i \}  \rfloor}{2^k}.
\end{align*}
Define
\begin{align*}
A_k \triangleq 2^k \left[\sum_{i=1}^m a_iX_i \right]_k, \quad 
B_k \triangleq 2^k \sum_{i=1}^m a_i \left[X_i \right]_k, \quad 
Z_k \triangleq \Bigg \lfloor \sum_{i=1}^m a_i \{ 2^kX_i \} \Bigg \rfloor.
\end{align*}
It is easy to see that $A_k, B_k, Z_k \in \mathbb{Z}^d$ and $A_k=B_k+Z_k$. Since $\{2^k X\} \in [0,1)^d$, each component of $Z_k$ takes integer values in the set $a_1[0,1)+\ldots+a_m[0,1)$ and hence $Z_k \in \mathcal{Z} \triangleq \{a,a+1,\ldots,b-1\}^d$, where $b \triangleq \sum_{i=1}^m a_i \mathbbm{1}_{\{a_i >0\}}$ and $a \triangleq \sum_{i=1}^m a_i \mathbbm{1}_{\{a_i <0\}}$. Hence $Z_k$ takes at most $(b-a)^d$ values, which is bounded for all $k$.

Next we describe the outline of the proof: 
\begin{enumerate}
	\item The goal is to prove $|H (A_k) - H(B_k)| \to 0$. Since $A_k=B_k+Z_k$, we have
	\begin{align}
H\left(A_k\right)-H\left(B_k\right)=I\left(Z_k;A_k\right)-I\left(Z_k;B_k\right). \label{eq:star}
\end{align}
Hence it suffices to show that both mutual informations vanish as $k\diverge$.

\item 
\prettyref{lmm:mutual2} proves $I\left(Z_k;B_k\right)\to 0$ based on the data processing inequality and \prettyref{lmm:intfracindep} which asserts that asymptotic independence between the integral part $\lfloor 2^kX \rfloor$ and the fractional part $ \{2^kX\}$, in the sense of vanishing mutual information. 
As will be evident in the proof of \prettyref{lmm:intfracindep}, this is a direct consequence of R\'enyi's result (\prettyref{lmm:lmmRenyi}).

\item 
Since $Z_k$ takes a bounded number of values, $I(Z_k;A_k)\to0$ is \emph{equivalent} to the total variation between $P_{Z_k,A_k}$ and $P_{Z_k}\otimes P_{A_k}$ vanishes, known as the $T$-information \cite{Csiszar96,PW14a}. 
By the triangle inequality and data processing inequality for the total variation, this objective is further reduced to proving the convergence of two pairs of conditional distributions in total variation: one is implied by Pinsker's inequality and \prettyref{lmm:mutual2}, and the other one follows from an elementary fact on the total variation between a pdf and a small shift of itself (\prettyref{lmm:TVshift}).
\prettyref{lmm:mutual1} contains the full proof; notably, the argument crucially depends on the assumption that $a_1,\ldots,a_m$ are relatively prime.


\end{enumerate}

We start with the following auxiliary result.
\begin{lemma}
\label{lmm:intfracindep}
Let $X$ be a $\left[0,1\right]^d$-valued continuous random variable such that both $h \left(X \right)$ and $H \left(\lfloor X \rfloor \right)$ are finite. Then
\begin{align*}
\lk I (\lfloor 2^kX \rfloor; \{2^kX\})=0.
\end{align*}
\end{lemma}
\begin{proof}
Since $X \in [0,1]^d$, we can write $X$ in terms of its binary expansion as:
 $$ X= \sum_{i \geq 1} X_i 2^{-i},X_i \in \{0,1\}^d.$$ 
In other words, $\lfloor 2^kX\rfloor =2^{k-1}X_1+\ldots+X_k$. Thus, $\lfloor 2^kX\rfloor$ and $\left(X_1,\ldots,X_k \right)$ are in a one-to-one correspondence and so are  $\{2^kX \}$ and $\left(X_{k+1},\ldots\right)$.
So, 
\begin{align*}
I (\lfloor 2^kX\rfloor; \{2^kX\} )&= I (X_1^k;X_{k+1}^{\infty}) \triangleq I\left(X_1,\ldots,X_k;X_{k+1},\ldots\right).
\end{align*}
Then $I\left(X_1^k;X_{k+1}^{\infty}\right)= \lim_{m \rightarrow \infty}I(X_1^k;X_{k+1}^{k+m})$ cf.~\cite[Section 3.5]{PW-it}.
Let $a_k \triangleq H\left(X_1^k\right) - dk\log2 - h\left(X\right)$.
Then \prettyref{lmm:lmmRenyi} implies $ \lim_{k \rightarrow \infty} a_k=0$. Hence for each $k,m \geq 1$, we have
\begin{align}
\nonumber
I (X_1^k;X_{k+1}^{k+m}) &= H (X_1^k )+H (X_{k+1}^{k+m} )-H (X_1^{k+m}) &\\
\nonumber
&=h(X)+d k \log2+ a_k - (h(X)+d(k+m)\log2+a_{k+m}) + H(X_{k+1}^{k+m})&\\
\nonumber
&=a_k-a_{k+m}+H(X_{k+1}^{k+m})- md\log2 &\\
& \leq a_k - a_{k+m}, \label{eq:binaryjustification}
\end{align}
where \eqref{eq:binaryjustification} follows from the fact that $X_{k+1}^{k+m}$ can take only $2^{md}$ values. 
Since $I(X_1^k;X_{k+1}^{k+m})\geq 0$, by \eqref{eq:binaryjustification}, sending $m \rightarrow \infty$ first and then $k \rightarrow \infty$ completes the proof.
\end{proof}

Recall that the total variation distance between probability distributions $\mu$ and $\nu$ is defined as:
\begin{align*}
\TV \left(\mu,\nu \right)\triangleq \sup_{F} | \mu(F)-\nu(F) | ,
\end{align*}
where the supremum is taken over all measurable sets $F$.

\begin{lemma}
\label{lmm:conditionalTV}
Let $X,Y,Z$ be random variables such that $Z= f \left(X\right)=f\left(Y\right)$, for some measurable function $f$. Then for any measurable $E$ such that $\prob {Z \in E }>0$,
$$
\TV \left( P_{X|Z \in E},P_{Y|Z\in E} \right) \leq \frac{\TV \left( P_X,P_Y \right)}{\prob{Z \in E }}. 
$$
\end{lemma}
\begin{proof}
For any measurable $F$,
\begin{align*}
\left| P_{X \in F |Z \in E} - P_{Y \in F|Z\in E} \right|= \frac{\left|\prob {X \in F, f \left(X\right) \in E}- \prob { Y \in F, f \left(Y\right) \in E}\right|}{\prob{Z \in E}} \leq \frac{\TV \left(P_X,P_Y \right) }{\prob {Z \in E }}.
\end{align*}
The claim now follows from taking supremum over all $F$.
\end{proof}

\begin{lemma}
\label{lmm:TVshift}
If $X$ is a $\mathbb{R}$-valued continuous random variable, then:
\begin{align*}
\TV ( P_X, P_{X+a}) \rightarrow 0 \mbox{ as } a \rightarrow 0.
\end{align*}
\end{lemma}
\begin{proof}
Let $f$ be the pdf of $X$. Since continuous functions with compact support are dense in $\mathcal{L}^1(\mathbb{R})$, for any $\varepsilon>0$, there exists a continuous and compactly supported function $g$ such that $\|f-g\|_1 < \frac{\varepsilon}{3}$. Because of the uniform continuity of continuous functions on compact sets, there exists a $\delta>0$ such that, whenever $|a| < \delta$, $\|g(\cdot+a)-g(\cdot)\|_1 < \frac{\varepsilon}{3}$. Hence $\|f(\cdot+a)-f(\cdot)\|_1< 2 \|f(\cdot)-g(\cdot)\|_1+ \|g(\cdot+a)-g(\cdot)\|_1 < \varepsilon$. Hence the claim follows.
\end{proof}

\begin{lemma}
\label{lmm:mutual2}
If $X_1,\ldots,X_m $ are independent $\left[0,1\right]^d$-valued continuous random variables such that both $h \left(X_j \right)$ and $H \left( \lfloor X_j \rfloor  \right)$ are finite for each $j \in \left[m\right]$, then
$$ 
\lim_{k \rightarrow \infty} I \left(Z_k;B_k\right)=0.
$$
\end{lemma}
\begin{proof}We have
\begin{align*}
I (Z_k;B_k ) &=  I \Big( \Big \lfloor \sum_{i=1}^m a_i \{ 2^kX_i \}  \Big \rfloor     ;  \sum_{i=1}^m a_i \lfloor 2^kX_i \rfloor  \Big) &\\
&=  I \Big( \Big \lfloor \sum_{i=1}^m a_i \{ 2^kX_i \}  \Big \rfloor     ; \Big \lfloor \sum_{i=1}^m a_i \lfloor 2^kX_i \rfloor \Big \rfloor  \Big) &\\
& \stackrel{(a)} \leq I \big(   a_1 \{ 2^kX_1 \},\ldots, a_m \{ 2^kX_m \}         ;  a_1 \lfloor 2^kX_1 \rfloor,\ldots,  a_m \lfloor 2^kX_m \rfloor \big) &\\
& \stackrel{(b)}= \sum_{i=1}^m I (  \{ 2^kX_i \} ;  \lfloor 2^kX_i \rfloor ),
\end{align*}
where $(a)$ follows from the data processing inequality and $(b)$ follows from the fact that $X_1,\ldots,X_m$ are independent. Applying \prettyref{lmm:intfracindep} to each $X_i$ finishes the proof.
\end{proof}

 In view of \eqref{eq:star}, \prettyref{lmm:implmm1} follows from \prettyref{lmm:mutual2} and the next lemma:
\begin{lemma}
\label{lmm:mutual1}
Under the assumptions of \prettyref{lmm:mutual2} and if $a_1,\ldots,a_m \in \Z$ are relatively prime,
$$
\lk I(Z_k;A_k) =0.
$$
\end{lemma}
%

\begin{proof}
Define the $T$-information between two random variables $X$ and $Y$ as follows:
$$
T(X;Y) \triangleq \TV(P_{XY}, P_X P_Y).
$$
By \cite[Proposition 12]{PW14a}, if a random variable $W$ takes values in a finite set $\mathcal{W}$, then
\begin{align}
I(W;Y) \leq  \log(|\mathcal{W}|-1)T(W;Y)+ h(T(W;Y)), \label{eq:yuwuineq}
\end{align}
where $h(x)=x \log \frac{1}{x} +(1-x) \log \frac{1}{1-x}$ is the binary entropy function.

Since $Z_k$ takes at most $\left(b-a\right)^d$ values, by \eqref{eq:yuwuineq}, it suffices to prove that $\lk T(Z_k;A_k)=0$. 
It is well-known that the uniform fine quantization error of a continuous random variable converges to the uniform distribution (see, \eg, \cite[Theorem 4.1]{JWW07}). Therefore $\{2^k X_i\}
          \xrightarrow{\mathcal{L}} \mathrm{Unif}[0,1]^d$ for each $i \in [m]$. Furthermore, since $X_i$ are independent, $Z_k=\lfloor \sum_{i=1}^m a_i
          \{2^k X_i\} \rfloor \xrightarrow{\mathcal{L}} \lfloor \sum_{i=1}^m a_i U_i \rfloor $ where $U_1,\dotsc,U_m$ are i.i.d. $\mathrm{Unif} [0,1]^d$ random variables.

Let $ \mathcal{Z}' \triangleq \{z \in \mathcal{Z}: \prob { \lfloor \sum_{i=1}^m a_i U_i \rfloor  =z } >0 \}$. Since $Z_k \xrightarrow{\mathcal{L}} \lfloor \sum_{i=1}^m a_i U_i \rfloor $, $\lk \prob{Z_k=z} >0$ for any $z \in \mathcal{Z}'$ and $\lk \prob {Z_k=z}=0$ for any $z \in \mathcal{Z} \backslash \mathcal{Z}'$.
Since
\begin{align*}
T(Z_k;A_k) & =  \sum_{z \in \mathcal{Z}} \prob{Z_k=z} \TV(P_{A_k},P_{A_k|Z_k=z})\\
& \leq  \sum_{z \in \mathcal{Z}'}\TV(P_{A_k},P_{A_k|Z_k=z})+ \sum_{z \in \mathcal{Z}\backslash \mathcal{Z}'} \prob {Z_k=z},
\end{align*}
it suffices to prove that $\TV(P_{A_k},P_{A_k|Z_k=z}) \to 0$ for any $z \in \mathcal{Z}'$.

Using the triangle inequality and the fact that $P_{A_k}=  \sum_{z' \in \mathcal{Z}} \prob{Z_k=z'} P_{A_k|Z_k=z'}$, we have
\begin{align*}
\TV(P_{A_k},P_{A_k|Z_k=z}) &\leq  \sum_{ z' \in \mathcal{Z}} \prob{Z_k=z'} \TV(P_{A_k|Z_k=z},P_{A_k|Z_k=z'})\\
& \leq  \sum_{z' \in \mathcal{Z}'} \TV(P_{A_k|Z_k=z},P_{A_k|Z_k=z'}) + \sum_{z \in \mathcal{Z}\backslash \mathcal{Z}'} \prob{Z_k=z}. 
\end{align*}
Thus it suffices to show that $\TV(P_{A_k|Z_k=z},P_{A_k|Z_k=z'}) \to 0$ for any $z,z' \in \mathcal{Z}'$. Since $A_k=B_k+Z_k$, we have
\begin{align}
\nonumber
\TV(P_{A_k|Z_k=z},P_{A_k|Z_k=z'}) & = \TV(P_{B_k+Z_k|Z_k=z},P_{B_k+Z_k|Z_k=z'}) \\
\nonumber
&= \TV(P_{B_k+z|Z_k=z},P_{B_k+z'|Z_k=z'}) \\
\nonumber
& \leq \TV(P_{B_k+z|Z_k=z},P_{B_k+z|Z_k=z'})+ \TV(P_{B_k+z|Z_k=z'},P_{B_k+z'|Z_k=z'})\\
& = \TV(P_{B_k|Z_k=z},P_{B_k|Z_k=z'})+ \TV(P_{B_k+z|Z_k=z'},P_{B_k+z'|Z_k=z'}). \label{eq:allweneed}
\end{align}
Thus it suffices to prove that each term on the right-hand side of \eqref{eq:allweneed} vanishes. For the first term, note that
\begin{align*}
\TV(P_{B_k|Z_k=z},P_{B_k|Z_k=z'}) \leq \TV(P_{B_k|Z_k=z},P_{B_k}) + \TV(P_{B_k|Z_k=z'},P_{B_k}),
\end{align*}
where $\TV(P_{B_k|Z_k=z},P_{B_k})\to 0$ for any  $z \in \mathcal{Z}'$ because, from the Pinsker's inequality,
\begin{align*}
I(Z_k;B_k) & = \sum_{z \in \mathcal{Z}} \prob{Z_k=z} D(P_{B_k}\|P_{B_k|Z_k=z})\\
& \geq 2 \sum_{z \in \mathcal{Z}} \prob{Z_k=z} \TV^2(P_{B_k},P_{B_k|Z_k=z})\\
& \geq 2 \prob{Z_k=z} \TV^2(P_{B_k},P_{B_k|Z_k=z}), 
\end{align*}
and $I(Z_k;B_k)\to 0$ by \prettyref{lmm:mutual2} and $\liminf_{k \rightarrow \infty}\prob{Z_k=z} >0$ for any $z \in \mathcal{Z}'$.

Thus it remains to prove the second term on the right-hand of \prettyref{eq:allweneed} vanishes for any $z,z' \in \mathcal{Z}'$.
Since $a_1,\ldots,a_m$ are relatively prime, for any $p \in \mathbb{Z}$, there exists $q_1,\ldots, q_m \in \mathbb{Z}$ such that $p= \sum_{i=1}^m a_i q_i$. Hence, for any $z,z' \in \Z^d$, there exists $b_1,\ldots,b_m \in \Z^d$ such that 
$$
z'-z =\sum_{i=1}^m a_i b_i.
$$
Then, 
\begin{align*}
B_k+\left(z'-z \right) &=  \sum_{i=1}^m a_i \lfloor 2^kX_i \rfloor + \sum_{i=1}^m a_i b_i= \sum_{i=1}^m a_i \Big \lfloor  2^k (X_i+\frac{b_i}{2^k})  \Big \rfloor.
\end{align*}
By definition, $Z_k=  \lfloor \sum_{i=1}^m a_i \{2^kX_i\}  \rfloor = \lfloor \sum_{i=1}^m a_i \{2^k (X_i+\frac{b_i}{2^k})\}  \rfloor $.
Consider the second term on the right-hand of \eqref{eq:allweneed}. We have
\begin{align*}
\TV ( P_{B_k+ z| Z_k=z'},  P_{B_k+z'| Z_k=z'}) &=
\TV ( P_{B_k+ \left(z'-z \right)| Z_k=z'},P_{B_k| Z_k=z'} )\\
&= \TV \big(  P_{ \sum_{i=1}^m a_i \lfloor  2^k (X_i+\frac{b_i}{2^k})   \rfloor|Z_k=z'   },P_{ \sum_{i=1}^m a_i \lfloor  2^k X_i   \rfloor |Z_k=z' } \big) &\\
& \stackrel{(a)}\leq  \TV ( P_{X_1+\frac{b_1}{2^k},\ldots,X_m+\frac{b_m}{2^k}|Z_k=z'}, P_{X_1,\ldots,X_m|Z_k=z'}) &\\
& \stackrel{(b)} \leq \frac{ 1}{\prob{Z_k=z'}}\TV ( P_{X_1+\frac{b_1}{2^k},\ldots,X_m+\frac{b_m}{2^k}}, P_{X_1,\ldots,X_m}) &\\
& \stackrel{(c)}\leq \frac{ 1}{\prob{Z_k=z'}}\sum_{i=1}^m \TV ( P_{X_i+\frac{b_i}{2^k}},P_{X_i}),
\end{align*}
where $(a)$ follows from the data processing inequality for total variation and $(b)$ follows from \prettyref{lmm:conditionalTV}, and $(c)$ follows from the independence of $X_1,\ldots,X_m$.
Letting $k \rightarrow \infty$ in view of \prettyref{lmm:TVshift} finishes the proof.
\end{proof}
 
\subsection{Proof of \prettyref{lmm:implmm2}}


\begin{proof}
Let $X_1,\ldots,X_m$ be independent and $\reals^d$-valued continuous random variables.
With out loss of generality, we may assume $a_i \neq 0$. For each $i \in \left[m\right]$, $\prob{X_i \in B_N^{(d)}} \xrightarrow{N \rightarrow \infty}1$.
Recall the conditional pdf notation \prettyref{eq:ftruncate}.
For $x \in \mathbb{R}^d$, we have
\begin{align}
f_{a_iX_i^{(N)}}(x)= \frac{1}{|a_i|} f_{X_i^{(N)}}\left(\frac{x}{a_i}\right)& = \dfrac{\frac{1}{|a_i|}f_{X_i}\left(\frac{x}{a_i}\right) \mathbbm{1}\left\{ \frac{x}{|a_i|} \in B_N^{(d)} \right\} }{\prob {X_i \in B_N^{(d)}}} = \dfrac{f_{a_iX_i}(x) \mathbbm{1}\left\{ \frac{x}{|a_i|} \in B_N^{(d)} \right\}  }{\prob {X_i \in B_N^{(d)}}}. \label{eq:truncpdf}
\end{align}
By the independence of the $X_i$'s, the pdf of $ \sum_{i=1}^m a_i X_i$ is given by:
\begin{align*}
g(z) & \triangleq f_{a_1X_1+ \ldots+a_m X_m}(z)\\
&= \int_{\mathbb{R}^{d}\times \cdots \times \mathbb{R}^d} f_{a_1 X_1} \left(x_1 \right)\ldots f_{a_m X_m}\left( z-x_1-\ldots-x_{m-1}\right)dx_1 \cdots dx_{m-1}.
\end{align*}
Similarly, in view of \eqref{eq:truncpdf}, the pdf of $ \sum_{i=1}^m a_i X_i^{(N)}$ is given by: 
\begin{align*}
g_N(z) &\triangleq f_{a_1X_1^{(N)}+\ldots+a_mX_m^{(N)}}(z) \\
&= \int f_{a_1 X_1^{(N)}} \left(x_1 \right)\ldots f_{a_m X_m^{(N)}}\left( z-x_1-\ldots-x_{m-1}\right)dx_1\ldots dx_{m-1}\\
&= \frac{1}{ \prod_{i=1}^m \prob {X_i \in B_N^{(d)} }} \cdot \int f_{a_1 X_1} \left(x_1 \right) \ldots f_{a_m X_m}\left(z-x_1-\ldots-x_{m-1}\right) \\
& \quad \hspace{2.6cm}\cdot\mathbbm{1}\left\{\frac{x}{|a_i|} \in B_N^{(d)} ,\ldots,\frac{z-x_1-\ldots-x_{m-1}}{|a_m|} \in B_N \right\} dx_1 \ldots dx_{m-1}. 
\end{align*}
Now taking the limit on both sides, we have $\lim_{N \rightarrow \infty} g_N(z) =g(z)$ a.e., which follows the dominated convergence theorem and the fact that $g(z)$ is finite a.e.

Next we prove that the differential entropy also converges. Let $N_0 \in \mathbb{N}$ be so large that 
\[
\prod_{i=1}^m \prob{X_i \in B_N^{(d)} } \geq \frac{1}{2}
\]
 for all $N \geq N_0$. Now, 
\begin{align*}
\left| h\pth{\sum_{j=1}^m a_j X_j}-h\pth{\sum_{j=1}^m a_j X_j^{(N)}} \right|
&= \left| \int_{\mathbb{R}^d} g \log \frac{1}{g} - \int_{\mathbb{R}^d} g_N \log \frac{1}{g_N} \right|\\
& \leq \int g_N \log \frac{g_N}{g}  +  \int \left| \left(g-g_N\right)\log \frac{1}{g} \right|\\
&= D \left(P_{\sum_{i=1}^m a_i X_i^{(N)}} \| P_{\sum_{i=1}^m a_i X_i}  \right) + \int \left| \left(g-g_N\right)\log g \right|\\
& \stackrel{(a)}\leq \sum_{i=1}^m D \left( P_{X_i^{(N)}}\|P_{X_i} \right) + \int \left| \left(g-g_N\right)\log g \right| \\
& \stackrel{(b)}=  \log \frac{1}{\prod_{i=1}^m \prob {X_i \in B_N^{(d)}}}+\int \left| \left(g-g_N\right)\log g \right| \\
& \stackrel{(c)} \to 0 \mbox{ as }{N \rightarrow \infty},
\end{align*}
where $(a)$ follows from the data processing inequality and $(b)$ is due to $D\left(P_{X|X\in E}\|P_{X}  \right)=\log \frac{1}{\prob {X \in E}}$, and $(c)$ follows from the dominated convergence theorem since $\left|\left(g-g_N\right)\log g \right| \leq 3g \left|\log g \right|$ for all $N \geq N_0$ and $\int g \left| \log g\right| < \infty$ by assumption. This completes the proof.
%
\end{proof}

\subsection{Proof of \prettyref{lmm:NoisyMutual}}
\label{sec:NoisyMutual}
\begin{proof}
In view of the concavity and shift-invariance of the differential entropy, without loss of generality, we may assume that $h(Z) < \infty$.
Since $U$ and $Z$ are independent, we have
\begin{align*}
I\pth{U;U+\varepsilon Z} &= h\pth{U+\varepsilon Z}- h \pth{U+\varepsilon Z |U}= h\pth{U+\varepsilon Z}-h(Z)-\log \varepsilon.
\end{align*}
 Hence it suffices to show that $\lim_{\varepsilon \rightarrow 0} I(U;U+\varepsilon Z) =H(U)$. Notice that $I(U;U+\varepsilon Z) \leq H(U)$ for all $\varepsilon$. On the other hand, $(U,U+\varepsilon Z) \xrightarrow{\mathcal{L}} (U,U)$ and $U+\varepsilon Z \xrightarrow{\mathcal{L}} U$ in distribution, by the continuity of the characteristic function. 
 By the weak lower semicontinuity of the divergence, we have
\begin{align*}
\liminf_{\varepsilon \rightarrow 0} I(U;U+\varepsilon Z) &=\liminf_{\varepsilon \rightarrow 0} D \pth{P_{U,U+\varepsilon Z} \| P_U P_{U+\varepsilon Z}}\\
& \geq D  \pth{P_{U,U} \| P_U P_U} = H(U),
\end{align*}
completing the proof.
\end{proof}

\subsection{Proof of \prettyref{lmm:Embedding}}
\label{sec:Embedding}
\begin{proof}
For any $\reals^d$-valued discrete random variable $U$, let $U_{[k]} \triangleq \pth{U_{(1)},\ldots,U_{(k)}}$, where $U_{(i)}$ are \iid copies of $U$. Thus $H \pth{U_{[k]}}=k H(U)$ and $\sum_{j=1}^m b_j (U_j)_{[k]}=\pth{\sum_{j=1}^m b_j U_j}_{[k]}$ for any $b_1,\ldots,b_m \in \reals$ and any discrete random variables $U_1,\ldots,U_m \in \reals^d$.

Let $U_1,\ldots,U_m$ be $\reals^d$-valued discrete random variables and $A=(a_{ij}) \in \reals^{n \times m}$. Let $\calU \subset \reals^d$ be a countable set such that $\sum_{i=1}^m a_{ij} U_j \in \calU $ for each $i \in [n]$. 
 Let $f_M:\reals^{d \times k} \to \reals^d$ be given by $f_M(x_1,\ldots,x_k)=\sum_{i=1}^m x_i M^i $ for $M>0$. Since for any $x=(x_1,\ldots,x_k)$  and $y=(y_1,\ldots,y_k)$ in $\calU^k$, there are at most $k$ values of $M$ such that $f_M(x)=f_M(y)$. Since $\calU^k$ is countable, $f_{M}$ is injective on $\calU^k$
 for all but at most countably many values of $M$. Fix an $M_0 >0$ such that $f_{M_0}$ is injective on $\calU^k$ and abbreviate $f_{M_0}$ by $f$. Let $U_j^{(k)}=f(\pth{U_j}_{[k]})$ for each $j \in [m]$. Thus, for each $i \in [n]$,
\begin{align*}
H \pth{\sum_{j=1}^m b_j U_j^{(k)}} &= H \pth{\sum_{j=1}^m a_{ij} f\pth{\pth{U_j}_{[k]}} } \stackrel{(a)}= H \pth{f \pth{\sum_{j=1}^m a_{ij} (U_j)_{[k]}}  }\\
&= H \pth{f \pth{\pth{\sum_{j=1}^m a_{ij} U_j}_{[k]} }  } \stackrel{(b)}= H \pth{ \pth{\sum_{j=1}^m a_{ij} U_j}_{[k]} }\\
&= k H \pth{ \sum_{j=1}^m a_{ij} U_j},
\end{align*}
where $(a)$ follows from the linearity of $f$ and $(b)$ follows form the injectivity of $f$ on $\calU^k$ and the invariance of Shannon entropy under injective maps.
\end{proof}

\section{Extensions to general groups}
\label{sec:groupstuff}
We now consider a more general version of \prettyref{thm:main}. To extend the notion of differential entropy to a more general setting, we need the following preliminaries. Let $G$ be a locally compact abelian group equipped with a Haar measure $\mu$. Let $X$ be a $G$-valued random variable whose distribution is absolutely continuous with respect to $\mu$. Following \cite{MK15}, we define the differential entropy of $X$ as:
$$
h\left(X\right) = \int f \log \frac{1}{f} d\mu = \Expect\qth{\log \frac{1}{f(X)}},
$$
where $f$ denotes the pdf of $X$ with respect to $\mu$. This extends both the Shannon entropy on $\integers^d$ (with $\mu$ being the counting measure) and the differential entropy on $\reals^d$ (with $\mu$ being the Lebesgue measure).

We now state a generalization of \prettyref{thm:main}, which holds for connected abelian Lie groups. Note that inequalities proved in \cite{MK15} using data processing inequalities hold for more general groups, such as locally compact groups on which Haar measures exist.

\begin{theorem}
\label{thm:linecircle}
Under the assumptions of \prettyref{thm:main}, suppose \eqref{eq:A} holds for any independent random variables $Z_1,\ldots,Z_m$ taking values in $\integers^d \times (\integers/2^k \integers)^n$ for any $k,d,n\in\naturals$. Then \eqref{eq:B} holds for any connected abelian Lie group $G'$ and independent $G'$-valued random variables $X_1,\ldots,X_m$.
\end{theorem}

We start by proving a special case of \prettyref{thm:linecircle} with $G$ being a finite cyclic group and $G'$ is the torus $\mathbb{T}^d$, where 
$\mathbb{T}$ denotes the unit circle in $\mathbb{C}$. 
\prettyref{thm:linecircle} then follows easily since any connected abelian Lie group is isomorphic to product of torus and Euclidean space.
We need the following preliminary fact relating the Haar measures and differential entropies of random variables taking values on isomorphic groups.

\begin{lemma}
\label{lmm:entropyinvar}
Let $\phi: G' \rightarrow G$ be a group isomorphism between abelian topological groups $(G,+)$ and $(G',+)$ and $\mu'$ be a Haar measure on $G'$. 
Then the pushforward measure\footnote{That is, $(\phi_{*}\mu') (B) = \mu' ( \phi^{-1} (B) )$ for any measurable subset $B$ of $G$.} $\mu=\phi_*\mu'$ is a Haar measure on $G$. Furthermore, for any $G$-valued continuous random variable $X$,
$$
h(X) =h \left(\phi^{-1}(X) \right).
$$
\end{lemma}
\begin{proof}
The first part is a standard exercise: For any measurable subset $A$ of $G$ and any $g \in G$, then
$$
\mu(g+A )= \mu' (\phi^{-1} ( g+A )) =\mu' ( \phi^{-1}( g ) +  \phi^{-1}( A) )=\mu' (\phi^{-1}(A))=\mu (A ),
$$
which follows the translation invariance of $\mu'$. Similarly, using the fact that $\phi^{-1}$ is a homeomorphism one can verify that $\mu$ is finite on all compacts as well as its inner and outer regularity.


If $f$ is the density function of $X$ with respect to the Haar measure $\phi_*\mu'$ on $G$, then $f \circ \phi$ is the pdf of $\phi^{-1}\left(X\right)$ with respect to the Haar measure $\mu'$ on $G'$. Hence,
\begin{align*}
h \left(X\right)
&=\int f \log \frac{1}{f} d (\phi_*\mu')\\
&=  \int f\circ \phi \log \frac{1}{f\circ \phi}d\mu\\
&= h \left(\phi^{-1}\left(X\right)\right). \qedhere
\end{align*}
\end{proof}

As an example, consider the group $(\reals^+, \times)$ of strictly positive real numbers with real multiplication, which is isomorphic to $(\reals,+)$ via $x\mapsto\log x$. Then for any $X\in(\reals^+, \times)$, its differential entropy is given by $h(X)=h(\log X)$, with the latter defined in the usual manner.

Define $\phi:[0,1)^n \rightarrow \mathbb{T}^n$ by $\phi(\theta_1,\ldots,\theta_n)=(e^{2 \pi i\theta_1},\ldots,e^{2\pi i\theta_n} )$. Let the Haar measure on $\mathbb{T}^n$ be the pushforward of Lebesgue measure under $\phi$. For $X \in \mathbb{T}^n$, let $\Theta= \phi^{-1}(X)$. Define the quantization operation of $X$ in terms of the angles
\begin{equation}
\left[X\right]_k \triangleq \phi \left( \frac{\lfloor 2^k \Theta \rfloor}{2^k}\right), \quad [\Theta]_k = \frac{\lfloor 2^k \Theta \rfloor}{2^k}.	
	\label{eq:quant-angle}
\end{equation}
Since $\phi$ is a bijection, $H \left([X]_k \right)=H \left( \lfloor 2^k \Theta \rfloor \right) $. We now prove \prettyref{thm:circles}.

\begin{theorem}
\label{thm:circles}
Under the assumptions of \prettyref{thm:main}, suppose \eqref{eq:A} holds for any cyclic group $G$-valued independent random variables $Z_1,\ldots,Z_m$. Then \eqref{eq:B} holds for any $\mathbb{T}^n$-valued independent random variables $X_1,\ldots,X_m$.
\end{theorem}
\begin{proof}
Let $X_1,\ldots,X_m$ be $\mathbb{T}^n$-valued continuous independent random variables. For each $i\in [m]$, let $\Theta_i=\phi^{-1}(X_i)$. Since $\lfloor 2^k \Theta_i \rfloor$ is $\mathbb{Z}_{2^k}$-valued and $\mathbb{Z}_{2^k}$ is a cyclic group under modulo $2^k$ addition, to prove \prettyref{thm:circles}, it suffices to prove the following:
\begin{align}
H \left( \left[X\right]_k \right) = kn \log2 + h\left(X\right)+o_k(1) \label{eq:step1}
\end{align}
for any $\mathbb{T}^n$-valued continuous random variable $X$, and 
\begin{align}
H \left( \left[ \sum_{i=1}^m a_i X_i \right]_k \right) = H \left( \sum_{i=1}^m a_i \left[X_i\right]_k \right) + o_k(1).\label{eq:step2}
\end{align}
Indeed, \eqref{eq:step1} follows from
\begin{align*}
H\left( \left[X\right]_k \right) &=H \left( [\Theta]_k \right)
\stackrel{(a)} = kn \log 2 + h\left(\Theta\right)+o_k(1)
\stackrel{(b)} = kn \log 2 + h \left(X \right)+o_k(1),
\end{align*}
where $(a)$ is by \prettyref{lmm:lmmRenyi} since $\Theta$ is a continuous $[0,1]$-valued random variable and $(b)$ is by \prettyref{lmm:entropyinvar}. To prove \eqref{eq:step2}, for each $i \in [m]$, let $\Theta_i=\phi^{-1}(X_i)$. Define
\begin{align*}
A_k &\triangleq  \left \lfloor 2^k  \sum_{i=1}^m  a_i \Theta_i \right \rfloor \Mod{2^k} , A_k'= \left \lfloor  2^k \sum_{i=1}^m  a_i \Theta_i \right \rfloor, \\
B_k &\triangleq \sum_{i=1}^m a_i \left \lfloor 2^k \Theta_i \right\rfloor \Mod{2^k}, B_k'= \sum_{i=1}^m a_i \left \lfloor 2^k \Theta_i \right\rfloor, \\
Z_k & \triangleq \left \lfloor \sum_{i=1}^m a_i \left\{2^k \Theta_i \right\}  \right \rfloor.
\end{align*}
Our aim is to prove that $H(A_k)-H(B_k)=o_k(1)$. Since $A_k'=B_k'+Z_k$, $A_k=B_k+Z_k \Mod{2^k}$. Also, $H(A_k)-H(B_k)=I(Z_k;A_k)-I(Z_k;B_k)$. Hence,
\begin{align*}
\left|H(A_k)-H(B_k)\right|  \leq I(Z_k;A_k)+I(Z_k;B_k)  \stackrel{(a)}\leq I(Z_k;A_k')+I(Z_k;B_k') \stackrel{(b)} \rightarrow 0 \mbox{ as } k \rightarrow \infty,
\end{align*}
where $(a)$ follows from the data processing inequality and $(b)$ follows from \prettyref{lmm:mutual2} and \prettyref{lmm:mutual1}. This completes the proof.
\end{proof} 

\begin{proof}[Proof of \prettyref{thm:linecircle}]
	The proof is almost identical to that of \prettyref{thm:circles}. By the structure theorem for connected abelian Lie groups (cf.~\eg~\cite[Corollary 1.4.21]{AM07}), $G'$
 is isomorphic to $\mathbb{R}^d \times \mathbb{T}^n$. By \prettyref{lmm:entropyinvar} and \prettyref{lmm:implmm2}, we only need to prove the theorem for $\left[0,1\right]^d \times \mathbb{T}^n$-valued random variables.
Along the lines of the proof of \prettyref{thm:circles}, it suffices to establish the counterparts of \eqref{eq:step1} for any $[0,1]^d \times \mathbb{T}^n$-valued continuous $X$, and \eqref{eq:step2} for any $[0,1]^d \times \mathbb{T}^n$-valued independent and continuous $X_1,\ldots,X_m$, where the quantization operations are defined componentwise by applying the usual uniform quantization \prettyref{eq:quant} to the real-valued components of $X$ and the angular quantization \prettyref{eq:quant-angle} to the $\torus^n$-component of $X$.
The argument is the same as that of \prettyref{thm:circles}, which we omit for conciseness.
\end{proof}

\section*{Acknowledgment}
The authors are grateful to Yury Polyanskiy and Mohamed-Ali Belabbas for discussions pertaining to \prettyref{thm:linecircle} and Mokshay Madiman for bringing \cite{Chan03} to our attention.
The authors thank Adriano Pastore for pointing out a mistake in the previous version and the reference \cite{JWW07}.
This work has been supported in part by NSF grants IIS-14-47879, CCF-14-23088 and CCF-15-27105 and the  Strategic Research Initiative on Big-Data Analytics of the College of Engineering at the University of Illinois.

\appendix
\section{Proof of \prettyref{prop:sharp}}
	\label{app:sharp}
\begin{proof}
The two equalities follows from \prettyref{thm:main} and \prettyref{thm:converse}.
Let $\alpha_n \triangleq \inf_{U \in \integers^n} \frac{H(U-U')-H(U)}{H(U+U')-H(U)}$. Clearly $\alpha_n \leq \alpha_1$ by the tensorization property of Shannon entropy.
On the other hand, given $U \in \integers^n$ and $U'$ its identical copy, using the same argument in the proof of \prettyref{lmm:Embedding}, there exists a linear embedding $f:\integers^n \to \integers$ that preserves the Shannon entropy of $U+U',U-U',U$ and $U'$. Hence
\begin{align*}
\frac{H(U-U')-H(U)}{H(U+U')-H(U)} &= \frac{H(f(U))-f(U'))-H(f(U))}{H(f(U)+f(U'))-H(f(U))}
\end{align*}
and $\alpha_1 \leq \alpha_n$. The result for the supremum follows from the same proof.
\end{proof}


\end{document}